\documentclass[11pt]{amsart}
\usepackage{graphicx, hyperref}

\newtheorem{thm}{Theorem}[section]
\newtheorem{cor}[thm]{Corollary}
\newtheorem{lem}[thm]{Lemma}
\newtheorem{prop}[thm]{Proposition}
\theoremstyle{definition}
\newtheorem{defn}[thm]{Definition}

\theoremstyle{remark}
\newtheorem{rem}[thm]{Remark}

\numberwithin{equation}{section}

\newcommand{\abs}[1]{\left\vert#1\right\vert}
\newcommand{\set}[1]{\left\{#1\right\}}
\newcommand{\Real}{\mathbb R}
\newcommand{\Natural}{\mathbb N}

\newcommand{\such}{\ | \ }
\newcommand{\prob}{\mathbb{P}}

\newcommand{\qprob}{\mathbb{Q}}
\newcommand{\expec}{\mathbb{E}}
\newcommand{\expecp}{\expec_\prob}
\newcommand{\expecq}{\expec_\qprob}
\newcommand{\var}{\mathbb{V} \mathsf{ar}}
\newcommand{\cov}{\mathbb{C} \mathsf{ov}}

\newcommand{\F}{\mathcal{F}}

\newcommand{\ud}{\mathrm d}
\newcommand{\inner}[2]{\left \langle #1 , #2 \right \rangle}

\newcommand{\argmax}{\operatorname{argmax}}

\newcommand{\secu}{S}

\newcommand{\pb}{p^{\mathsf{r}}}
\newcommand{\qb}{q^{\mathsf{r}}}

\newcommand{\pare}[1]{\left(#1\right)}
\newcommand{\bra}[1]{\left[#1\right]}
\newcommand{\dbra}[1]{[\kern-0.15em[ #1 ]\kern-0.15em]}
\newcommand{\dbraco}[1]{[\kern-0.15em[ #1 [\kern-0.15em[}
\newcommand{\dbraoc}[1]{]\kern-0.15em] #1 ]\kern-0.15em]}

\newcommand{\Lb}{\mathbb{L}}

\newcommand{\U}{\mathbb{U}}
\newcommand{\V}{\mathbb{V}}

\newcommand{\dfn}{\mathrel{\mathop:}=}

\newcommand{\iii}{i \in I}
\newcommand{\jii}{j \in I}
\newcommand{\iij}{i \in J}

\newcommand{\sumi}{\sum_{\iii}}

\newcommand{\rto}{\theta}
\newcommand{\rtoi}{\theta_i}
\newcommand{\rtoni}{\theta_{-i}}
\newcommand{\kb}{k^{\mathsf{r}}}
\newcommand{\rb}{\rtoi^{\mathsf{r}}}

\newcommand{\rg}{\theta^*}
\newcommand{\kg}{k^*}
\newcommand{\rgi}{\rtoi^{*}}
\newcommand{\rgni}{\rtoni^{*}}

\newcommand{\pg}{p^{*}}
\newcommand{\qgi}{q_i^{*}}
\newcommand{\qg}{q^{*}}
\newcommand{\DU}{\mathsf{DU}}

\newcommand{\pp}{\widehat{p}}
\newcommand{\pq}{\widehat{q}}

\begin{document}

\title[Effective risk aversion in thin risk-sharing markets]{Effective risk aversion in thin risk-sharing markets}%

\author{Michail Anthropelos}
\address{Michail Anthropelos, Department of Banking and Financial Management, University of Piraeus.}
\email{anthropel@unipi.gr}

\author{Constantinos Kardaras}
\address{Constantinos Kardaras, Statistics Department, London School of Economics.}
\email{k.kardaras@lse.ac.uk}

\author{Georgios Vichos}
\address{Georgios Vichos, Statistics Department, London School of Economics.}
\email{g.vichos@lse.ac.uk}

\thanks{The authors are grateful to Jak\v{s}a Cvitani\'{c} for discussions that inspired this work. Many thanks also go to Andrea Buffa, Paolo Guasoni, Scott Robertson and the participants of academic seminars in Dublin City University, University of Michigan and the Questrom School of Business of Boston University.}

\date{\today}%
\begin{abstract}
We consider thin incomplete financial markets, where traders with heterogeneous preferences and risk exposures have motive to behave strategically regarding the demand schedules they submit, thereby impacting prices and allocations. We argue that traders relatively more exposed to market risk tend to submit more elastic demand functions. Noncompetitive equilibrium prices and allocations result as an outcome of a game among traders. General sufficient conditions for existence and uniqueness of such equilibrium are provided, with an extensive analysis of two-trader transactions. Even though strategic behaviour causes inefficient social allocations, traders with sufficiently high risk tolerance and/or large initial exposure to market risk obtain more utility gain in the noncompetitive equilibrium, when compared to the competitive one.       
\end{abstract}

\maketitle


\section*{Introduction}

It has been widely recognised that many financial markets are dominated by a relatively small number of large investors, whose actions heavily influence prices and allocations of tradeable securities---see, among others, discussions in \cite{BluKei12, GibSinYer03, RosWer15}. While such market impact has been observed even in large exchanges like NYSE (see \cite{KeiMad95, KeiMad96, MadChe97} and the more recent empirical study \cite{AllMatMat17}), it is especially in over-the-counter (OTC) transactions that the assumption of a competitive market structure is problematic. The majority of OTC markets involve relatively few participants; therefore, even if all information is public, equilibrium forms in a noncompetitive manner. Such financial markets with an oligopolistic structure are usually characterised as \emph{thin} (see \cite{RosWer08} for a related reviewing discussion).

The main reason for trading between risk averse traders with common information and beliefs is the heterogeneity of their endowments---see, for instance, related discussion in \cite{BarElKar05, JouSchaTou08}. Trading securities that are correlated with traders' endowments may be mutually beneficial in sharing the traders' risky positions---see, among others, \cite{AnthZit10, Rob17}. In a standard Walrasian uniform-price auction model, traders submit demand schedules on the tradeable securities and the market clears at the prices resulting in zero aggregate submitted demand; and since demand depends on traders' characteristics such as their risk exposure and risk aversion, the same is true for the equilibrium prices and allocation. Whereas traders' exposures to market risk (i.e., uncertainty in the tradeable securities) may be considered as public knowledge, their risk aversion is \emph{subjective} and should be regarded as private information. In the realm of thin financial markets, traders may have motive to act strategically and submit demand schedules with \emph{different} elasticity than the one reflecting their risk aversion. The goal of this paper is to model such strategic behaviour and highlight some of its economic insights.

\subsection*{Model description and main contributions}
We develop a model of a one-shot transaction on a given collection of risky tradeable securities, under common information on the probabilistic nature of their payoffs. Traders possess and exploit a potential to impact the market's equilibrium. We adopt the setting of CARA preferences and normally distributed payoffs, also appearing in \cite{Kyl89, RosWer15, Vay99}, with traders assumed heterogeneous with respect to risk tolerance (defined as the reciprocal of risk aversion) and initial risky positions. In contrast to the majority of related literature, we do not assume that traders' endowments belong to the span of the tradeable securities, leading to market incompleteness.   

Similarly to the models in \cite{Kyl89, Vay99, Viv11}, the market operates as a uniform-price auction where traders submit demand functions on the tradeable securities, with equilibrium occurring at the price vector that clears the market. When traders do not act strategically, the market structure is competitive and the equilibrium price-allocation is induced by traders' \emph{true} demand functions. However, as has been pointed out previously, such competitive structure is not suitable for thin markets, and the way traders behave depends in principle on the risk exposure and risk tolerance of their counter-parties. In a CARA-normal setting, demand functions are linear with downward slope and their elasticities coinciding with the traders' risk tolerance. Traders recognise their ability to influence the equilibrium transaction, and may submit demand with different elasticity than the one reflecting their risk tolerance.  We formulate a best-response problem, according to which traders submit demand functions aiming at individual utility maximisation, with strategic choices parametrised by the elasticity of the submitted demand. This forms a noncompetitive market scheme, where the Bayesian Nash equilibrium is the fixed point of traders' best responses.

In any non-trivial case, traders have motive to submit demand with different elasticity than their risk tolerance. The main determining factor of traders' best response is their pre-transaction  \emph{projected beta}, defined as the beta (in terms of the Capital Asset Pricing Model) of the projection of the trader's risky position onto the linear space generated by the securities. In the special case where the traders' positions belong to the span of tradeable securities, 
projected and actual betas coincide. Following classical literature, traders' projected betas (hereafter simply called betas) measure their exposure to market risk. In terms of risk sharing, we distinguish traders to those who increase or decrease their beta through the transaction.

It is shown that traders submit demand corresponding to higher risk tolerance if and only if they reduce their market risk exposure through trading. The economic insight of this strategic behaviour is simple: traders with relatively higher initial exposure to market risk pay a \textit{risk premium} to their counter-parties in order to reduce their beta. Submitting more elastic demand has two main effects. Firstly, the post-transaction reduction of beta is smaller, since more elastic demand implies higher relative risk tolerance and hence higher post-transaction exposure to market risk, as the trader appears willing to keep a more risky position. Secondly, the risk premium that is paid is also lower. As it turns out, the effect of premium reduction overtakes the sub-optimal reduction of market risk exposure. In order to obtain intuition on this, consider the impact of the other traders' status on an individual trader's actions. Large pre-transaction beta for a specific trader implies low aggregate beta for other traders. Acting in a more risk tolerant way, by submitting more elastic demand, a trader essentially exploits this low aggregate exposure to market risk of the counter-parties, and in fact decreases the premium that they ask in order to undertake more market risk. 

On the contrary, traders who undertake market risk in exchange for a risk premium, i.e., those with low pre-transaction beta, have motive to submit less elastic demand. Not only does such a strategy result in less undertaken market risk, it also takes advantage of the large aggregate counter-parties' beta, increasing the premium received in order to offset their demand.   

Continuing this line of argument, traders with overexposed to market risk, with pre-transaction beta sufficiently higher than one, tend to behave as risk neutral, even though their actual risk aversion parameter is strictly positive. In such a case, the trader takes over the whole market risk, reducing the post-transaction beta of their position to one. At the same time, the other traders are willing to offset such transaction since it makes their post-transaction beta equal to zero (i.e., becoming market-neutral); for this reason, they reduce the required risk premium. On the other hand, traders with pre-transaction beta less than or equal to $-1$ submit extremely inelastic demand functions, implying zero risk tolerance, appearing willing to become market neutral. Again, other traders are eager to offset the transaction, since at this regime their aggregate pre-transaction beta is relatively large, and selling market risk is a very effective hedging transaction.

We discover two regimes of noncompetitive equilibrium. When one of the trader's pre-transaction beta is sufficiently large, there exists a unique linear equilibrium which is \emph{extreme}, in the sense that the market-overexposed trader behaves as being risk neutral and at equilibrium undertakes all market risk. Such extreme Nash equilibrium results in market-neutral portfolios for all other traders, while securities are priced in a risk-neutral manner. In any other ``\emph{non-extreme}'' case, noncompetitive equilibria solve a coupled system of quadratic equations, which admits a unique solution under the mild---and rather realistic---assumption that at most one of the traders may have pre-transaction beta greater than one. We provide an efficient constructive proof of the latter fact, which can be used to numerically obtain the unique linear equilibrium given an arbitrary number of traders. 

The two-trader case is of special interest, mainly because the large majority of risk-sharing transactions are bilateral between large institutions and/or their clients or brokers; related discussions and statistics are provided in \cite{BabHu16, Bab16, DufSchVui15, Zaw13, HedMad15}. We obtain explicit expressions for two-trader price-allocation noncompetitive equilibria, which allow us to analyse further the model's economic insight. Noncompetitive and competitive equilibria coincide if and only if the competitive equilibrium transaction is null, in that the initial allocation is already Pareto-optimal. In any other case, for both traders the elasticity of submitted demands in such thin market deviates from the one utilising their risk tolerances. As emphasised above, the crucial factor is the traders' pre-transaction beta. For non-extreme equilibria we have the following synoptic relationship: 
\begin{center}
\emph{true elasticity $<$ equilibrium elasticity $\quad\Leftrightarrow\quad$ post-transaction beta $<$ pre-transaction beta}.
\end{center}
Even if traders have common risk tolerance, deviations between their endowment will make them behave heterogeneously. For a trader with higher (resp., lower) beta, who reduces (resp., increases) market risk through the transaction, the equilibrium elasticity reflects more (resp., less) risk tolerance. One could argue, therefore, that in thin financial markets the assumption of effectively homogeneous risk-averse traders is problematic, since it essentially implies that traders ignore their ability to impact the transaction. 

In the context of strategic behaviour, equilibrium prices and allocations are generally impacted. In the two-trader case, the volume in noncompetitive equilibrium is always lower than in the competitive one. More precisely, it is shown that the post-transaction beta after Nash equilibrium is---interestingly enough---the midpoint between the trader's pre-transaction beta and the beta after the competitive transaction. This implies a loss of social efficiency, in the sense that the total utility in noncompetitive equilibrium is reduced when compared to the competitive one. However, such loss of total utility does not always transfer to the individual level. In fact, it follows from the analysis of the bilateral game that the noncompetitive equilibrium is beneficial in terms of utility gain for two types of traders: those with sufficiently high pre-transaction beta, and those with sufficiently high risk tolerance. Such findings in noncompetitive markets are consistent with results in \cite{Anth17} and \cite{AnthKar17}. (A result in that spirit also appears in \cite{MalRos14}; namely it is shown that, when the market is centralised, less risk averse agents have greater price impact.)

As a final point, and as mentioned above, our model allows for incompleteness, and we study its effect in noncompetitive risk-sharing transaction. Based on the two-trader game, we show that traders who benefit from the noncompetitive market setting (i.e., those with high risk tolerance and/or high exposure to market risk) have their utility gains reduced by the fact that endowments are not securitised, highlighting the importance of completeness especially for large traders that prefer thin markets for sharing risk.

\subsection*{Connections with related literature}

The present paper contributes to the large literature on imperfectly competitive financial markets. Based on the seminal works on Nash equilibrium in supply/demand functions of \cite{KleMey89} and \cite{Kyl89}, most models of noncompetitive markets consider strategically acting agents, whose set of choices corresponds to demand schedules submitted to the transaction. Frequently, the departure from competitive structure stems from informational asymmetry; such is the case in \cite{Bac92, BlaCaoWil00, Kyl89, KylObiWan14}, where agents are categorised as informed, uniformed and noisy. Even without existing risky positions, asymmetric information gives rise to mutually beneficial trading opportunities among traders, who submit demand schedules based on the responses of their counter-parties. Another potential source of noncompetitiveness comes via exogenously imposing asymmetry on the bargaining power among market participants. Bilateral OTC transactions between agents with different bargaining power are modelled in \cite{DuffGarPed07}; in \cite{LiuWan16}, it is market makers who possess market power and optimally adjust bid-ask spreads based on submitted orders by informed and uniformed investors. (See the references in \cite{LiuWan16} for alternative models of strategic market makers.) Exogenously imposed differences on market power are also present in \cite{BruPed05}, where traders are divided into price-takers and predatory ones, the latter strategically exploiting the liquidity needs of their counter-parties.      

In contrast to the above, our model assumes symmetry for traders' market power; noncompetitiveness stems solely from the fact there is a small number of acting traders, each of whom can buy or sell the tradeable securities and has the ability to affect the risk-sharing transaction.\footnote{Symmetric games in an oligopolistic market of goods (rather than securities with stochastic payoffs) have also been studied in the seminal work of \cite{KleMey89} and in the more recent papers \cite{Viv11} and \cite{Wer11}. The main structural difference between these market models and ours is that players therein (i.e., firms) can take only the seller's side, while the buyer's side (i.e., the demand for the goods) is essentially exogenous. Additionally, the fact that the tradeable asset is a good creates further technical and economic deviations---for instance, the role of risk exposure is essentially played by the cost function, the price can not be negative, etc. The model in \cite{KleMey89} imposes randomness on demand, whereas \cite{Viv11} considers random suppliers' cost and private information status. On the other hand, the model of market power in \cite{Wer11} is based on the same setting of price impact as in \cite{RosWer15} and \cite{MalRos14}.} The market here is assumed to be oligopolistic, without any form of exogenous frictions or asymmetries.

Market models close to ours considered by other authors include \cite{MalRos14, RosWer15, Vay99}. In \cite{MalRos14, RosWer15, Vay99}, and similarly to the present work, traders submit demand in a noncompetitive market setting by taking into account the impact of their orders on the equilibrium. The main difference with our demand-game, when compared to the one-shot market of \cite{RosWer15,Vay99} and the centralised market of \cite{MalRos14}, is the set of traders' strategic choices. More precisely, in these works a trader's price impact is identified as the slope of the submitted aggregate demand of the rest of the traders. Traders estimate (correctly at equilibrium) their price impact and respond by submitting demand schedules aiming for maximising their own utility. In particular, the set of strategic choices consists of the slope of the submitted demand, and equilibrium arises as the fixed point of the traders' price impacts. In our model, we keep the linear equilibrium structure of demand functions and parametrise the set of traders' strategical choice to the submitted \textit{elasticity}, and equilibrium is formed simply at the price where aggregate submitted demand is zero. In this way, each trader responds to the whole demand function of other traders, and not just the slope. This is a crucial trading feature motivated by the benefits of risk sharing, since the intercept point of the demand function corresponds to the traders' exposure to market risk (the correlation of traders' endowment with the tradeable assets). The difference becomes pronounced in the very special case of a single tradeable security, where traders' price impacts of \cite{RosWer15} and \cite{MalRos14} can be seen as the reciprocal of their risk aversion. In \cite{RosWer15}, the so-called equilibrium \textit{effective risk aversion}---that is, the risk aversion that is reflected by the equilibrium submitted demands---depends only on the number of traders (as well as a couple of other quantities that we do not use in our model: interest rate and number of allowable trades until the end of each trading round). In particular, heterogeneity of initial risky endowments is not addressed: even with different initial positions at each period, traders do not take into account their counter-parties' exposure to market risk. Our demand-game may be more appropriate for thin risk-sharing transactions, since it endogenously highlights the importance of traders' initial positions for their strategic behaviour.

Another important trait of our model is that it can be applied to the practically important two-trader case, while the models of  \cite{RosWer15}, \cite{MalRos14} and \cite{Vay99} are ill-posed for bilateral transactions. As already mentioned, bilateral transactions are significant part of thin market models, since the majority of the OTC risk-sharing transactions consist of only two counter-parties. Existence of a two-agent Bayesian Nash equilibrium exists under mild assumptions in the model of \cite{RosWer12}; however, agents there have private valuations on the tradeable securities. 

Further to what was pointed out above, our model allows market incompleteness: tradeable securities do not necessarily span the traders' endowments. We are thus able to generalise the discussion on thin markets and deviations of noncompetitive equilibria from competitive ones in the more realistic framework where traders' endowments are nor securitised  neither replicable. 

Finally, models of thin risk-sharing markets, albeit with different set of strategic choices, have been considered in \cite{Anth17} and \cite{AnthKar17}. In \cite{Anth17}, traders choose the \emph{endowment} submitted for sharing, and a game on agents' linear demand is formed; in contrast with the present paper, agents in \cite{Anth17} choose the intercept of the demand function instead of its elasticity. In \cite{AnthKar17}, traders strategically submit \emph{probabilistic beliefs}, and the model is ``inefficiently complete'', as securities are endogenously designed by heterogeneous traders in order to share their risky endowments.

\subsection*{Structure of the paper}

Section \ref{sec:setup} introduces the market model and competitive equilibrium, where traders do not act strategically. Section \ref{sec:best_response} introduces, solves and discusses the individual trader's best response problem. Noncompetitive equilibrium is introduced in Section \ref{sec.Nash}; general conditions ensuring existence and uniqueness of Nash equilibrium are provided in \S \ref{subsec:More_than_two}, conditions for so-called extreme equilibrium are addressed in \S \ref{subsec:extreme_case}. The two-trader game is extensively analysed in Section \ref{sec.bilateral}. The proof of the main Theorem \ref{thm:ex_and_un} is presented in Appendix \ref{sec:appe}.

\section{Model Set-Up}\label{sec:setup}

We work on a probability space $(\Omega, \, \mathcal{F}, \, \prob)$, and denote by $\Lb^0 \equiv \Lb^0(\Omega, \mathcal{F},\prob)$ the class of all $\F$-measurable random variables, identified modulo $\prob$-a.s.~equality.

\subsection{Agents and preferences}

We consider a market of $n + 1$ economic traders, where $n \in \Natural = \set{1, 2, \ldots}$; for concreteness, define the index set $I = \set{0, \ldots, n}$. Traders are assumed risk averse and derive utility only from future consumption of a num\'eraire at the end of a single period, where all uncertainty is resolved. To simplify the analysis we assume that all considered security payoffs are expressed in units of the num\'eraire, which implies that future deterministic amounts have the same present value for
the traders. Each trader $\iii$ carries a risky future payoff in units of the num\'eraire, which is called (random) \emph{endowment}, and denoted by $E_i$. The endowment $E_i \in \Lb^0$ denotes the existing risky portfolio of trader $\iii$, and is not necessarily securitised or tradeable. We define the aggregate endowment $E_I \dfn \sumi E_i$, and set $E \equiv (E_i)_{\iii}$ to be the vector of traders' endowments.

The preference structure of traders is numerically represented by the functionals
\begin{equation}\label{eq:utility_functional}
\Lb^0 \ni X \mapsto \U_i(X) \dfn - \delta_i \log \expec\bra{\exp \pare{- X /
		\delta_i}} \in [- \infty, \infty),
\end{equation}
where $\delta_i \in (0, \infty)$ is the risk tolerance of trader $\iii$. Note that $\U_i(X)$ corresponds to the certainty equivalent of potential future random outcome $X$, when trader $\iii$ has risk preferences with constant absolute risk aversion (CARA) equal to $1 / \delta_i$. It is important to point out that functional $\U_i(\cdot)$ also measures wealth in num\'eraire units and hence can be used for comparison among different traders (and equilibria).  We also define the \emph{aggregate risk tolerance} $\delta_I \dfn \sumi \delta_i$, as well as the \emph{relative risk tolerance} $\lambda_i \dfn \delta_i / \delta_I$ of trader $\iii$. Note that $\lambda_I \equiv \sumi \lambda_i = 1$. Following standard practice, we shall use subscript ``$-i$'' to denote aggregate quantities of all traders except trader $\iii$; for example, $\delta_{-i} \dfn \delta_I - \delta_i$ and $\lambda_{-i} \dfn 1 - \lambda_i$, for all $\iii$.

\subsection{Securities and demand}

In the market there exist a finite number of tradeable securities indexed by the non-empty set $K$, with payoffs denoted by $\secu \equiv (\secu_k)_{k \in K} \in(\Lb^0)^K$. The demand function $Q_i$ of trader $\iii$ on the vector $\secu$ of securities is given by
\[
Q_i(p)\dfn \underset{q \in \Real^K}{\argmax} \ \U_i (E_i + \inner{q}{S - p}), \quad p \in \Real^K.
\]
Here, and in the sequel, $\inner{\cdot}{\cdot}$ will denote standard inner product on the Euclidean space $\Real^K$.

We follow a classic model of standard literature (e.g.~\cite{Kyl89, RosWer15, Vay99} and \cite{Viv11}) and assume that the joint law of $(E, \secu)$ is Gaussian. Since traders' endowments do not necessarily belong to the span of $\secu$, the market is incomplete. Note also that endowments are not assumed independent of $\secu$, or independent of each other. Since only securities in random vector $\secu$ are tradeable, we identify \emph{market risk} with the variance-covariance matrix of $\secu$, denoted by
\[
C \dfn \cov(\secu, \secu).
\]
In the sequel will impose the standing assumption that $C$ has full rank. Additionally, for notational convenience, we shall assume that
\[
\expec[\secu_k] = 0, \quad \forall k \in K.
\]
Due to the cash-invariance of the traders' certainty equivalent, the latter assumption does not entail any loss of generality, as we can normalise tradeable securities to be $\secu - \expec[\secu]$. Straightforward computations give
\begin{align*}
\U_i \pare{E_i + \inner{q}{\secu - p}} &= - \delta_i \log \expec \bra{\exp \pare{- (E_i + \inner{q}{\secu - p}) /
		\delta_i}} \\
&= \expec \bra{E_i} - \inner{q}{p} - \frac{1}{2 \delta_i} \var \bra{E_i + \inner{q}{\secu}} \\
&= \expec \bra{E_i} - \frac{1}{2 \delta_i} \var \bra{E_i} - \frac{1}{2 \delta_i} \inner{q}{C q} - \inner{q}{p + \frac{1}{\delta_i} \cov (E_i, \secu)}. 
\end{align*}
We also define the following quantities
\[
u_i \dfn \expec \bra{E_i} - \frac{1}{2 \delta_i} \var \bra{E_i} \equiv \U_i(E_i),\] 
and, for each $\iii$,
\[
a_i \dfn C^{-1} \cov (E_i, \secu), \quad \text{and} \quad a_{-i}\dfn a_I - a_i,
\]
where
\[
a_I \dfn \sum_{\iii} a_i.
\]
Then, it follows that
\[
\U_i \pare{E_i + \inner{q}{\secu - p}} = u_i - \frac{1}{\delta_i}\inner{q}{C a_i} - \frac{1}{2 \delta_i} \inner{q}{C q} - \inner{p}{q},
\]
from which we readily obtain that the demand function of trader $\iii$, given by
\begin{equation}\label{eq.demand}
\Real^K \ni p \mapsto Q_i(p) = - a_i - \delta_i C^{-1} p, \quad \iii,
\end{equation}
is downward-sloping linear. The risk tolerance $\delta_i \in (0, \infty)$ could be considered as the \textit{elasticity} of the demand function of trader $\iii$, with higher $\delta_i$ implying more elastic demand. Furthermore, $a_i \in \Real^K$ gives the correlation of the tradeable securities with the endowment of trader $\iii$, and plays the role of the intercept point of the affine demand function \eqref{eq.demand}. According to \eqref{eq.demand}, when prices of all securities equal zero, the sign of each element of $a_i$ indicates whether trader $\iii$ has incentive to buy (when negative) or sell (when positive) the corresponding security.  

\subsection{Competitive equilibrium}

While our focus will be on noncompetitive equilibrium, we first define competitive equilibrium of our market, to be later used and discussed as a benchmark for comparison, similarly as in \cite{Vay99} and \cite{Viv11}.  Trading the securities represented by $\secu$ without applying any strategic behaviour (i.e., by assuming a \emph{price-taking} mechanism), the traders reach a competitive equilibrium: prices are determined where the traders' aggregate demand equals zero.  

\begin{defn} 
The vector $\pp \in \Real^K$ is called \textbf{competitive equilibrium prices} if 
\[
\sum_{\iii} Q_i (\pp)=0. 
\]
The corresponding allocation $(\pq_i)_{\iii} \in \Real^{K \times I}$ defined via $\pq_i = Q_i(\pp)$ for all $\iii$ will be called a \textbf{competitive equilibrium allocation} associated to (competitive equilibrium) prices $\pp \in \Real^K$.
\end{defn}

Elementary algebra gives the following result.

\begin{prop}\label{prop:comp_equilibrium}
	There exists a unique competitive equilibrium price $\pp$ given by
	\begin{equation}\label{eq:p_comp_equil}
	\pp = - \frac{1}{\delta_I} C a_I,
	\end{equation}
	with associated competitive equilibrium allocations given by
	\begin{equation}\label{eq:q_comp_equil}
	\pq_i = \lambda_i a_I - a_i, \quad \iii.
	\end{equation}
\end{prop}

\begin{rem}\label{rem:D}
	For $\iii$, $D_i \dfn \inner{a_i}{S}$ is the projection of the endowment $E_i$ onto the linear span of the tradeable security vector $\secu \equiv (\secu_k)_{k \in K}$. At competitive equilibrium, the position of trader $\iii$, net the price paid, is
	\[
	\inner{\pq_i}{S - \pp} = \inner{\lambda_i a_I - a_i}{S} + \frac{1}{\delta_I} \inner{\lambda_i a_I - a_i}{Ca} = \lambda_i D_I - D_i - \expecq \bra{\lambda_i D_I - D_i}, \quad \iii.
	\]
	where $\qprob$ is given through $\ud \qprob / \ud \prob = \exp(-E_I / \delta_I)/\expec_{\prob}\bra{\exp(-E_I / \delta_I)}$, where $D_I \dfn \sum_{\iii}D_i$. In the case where the linear span of the securities equals the linear span of the endowments, it holds that $D_i = E_i - \expecp \bra{E_i}$, for all $\iii$. Then, the competitive equilibrium coincides with the complete-market Arrow-Debreu risk-sharing equilibrium---see, among others, \cite{Bor62, Buhl84} or \cite[Chapters 2 and 3]{MagQui02}.
\end{rem}

\begin{rem} \label{rem:uninteresting_1}
A very special---and as shall be discussed, trivial---situation arises when $a_I = 0$, i.e., when $\cov (E_I, \secu_k) = 0$ holds for every $k \in K$, where we recall that $E_I \dfn \sumi E_i$. In words, $a_I = 0$ means that the total endowment $E_I$ is independent of the spanned subspace of the securities. In this case, in the setting of Proposition \ref{prop:comp_equilibrium}, competitive equilibrium prices of the securities are zero, and $\pq_i = - a_i$. It follows that, in competitive equilibrium, traders simply rid themselves of the hedgeable part of their endowment at zero prices, and end up after the transaction with the part that is independent of the securities. (In this respect, recall the previous Remark \ref{rem:D}.)
\end{rem}

Given that the case $a_I = 0$ is covered by Remark \ref{rem:uninteresting_1} above, we shall assume tacitly in the sequel that $a_I \neq 0$. (The only point where we return to the case $a_I = 0$ is at Remarks \ref{rem:uninteresting_2} and \ref{rem:uninteresting_3}.) When $a_I \neq 0$, we define the following parameters, which will turn out to be crucial for our analysis:
\begin{equation}\label{eq.beta}
\beta_i \dfn \frac{\cov(E_I,\secu)C^{-1}\cov(E_i,\secu)}{\cov(E_I,\secu)C^{-1}\cov(E_I,\secu)}=\frac{\inner{a_I}{Ca_i}}{\inner{a_I}{C a_I}},\quad \iii.
\end{equation}
Note that
\[
\beta_I \equiv \sum_{\iii} \beta_i = 1.
\]
When the traders' endowments are tradeable, i.e., when the endowment vector $E_i$ belongs in the linear span of $(\secu_k)_{k \in K}$ for all $\iii$, then $\beta_i$ literally coincides with the beta of the $i$th endowment, in the terminology of the Capital Asset Pricing Model. In general, $\beta_i$ should be considered as a ``\emph{projected beta}'' of the $i$th endowment onto the space of tradeable securities; as stated in the introduction, it shall be called simply (pre-transaction) \emph{beta} in the sequel. Consistently to classical theory, betas shall measure the level of \textit{exposure to market risk} of each trader before and after the equilibrium transaction.

Both equilibrium prices and allocations strongly depend on the traders' heterogeneity. After the competitive transaction, the position of trader $\iii$ is $E_i + \inner{\pq_i}{\secu - \pp}$, and one may immediately calculate the post-transaction beta of the position to be equal to $\lambda_i$. Hence, at competitive risk sharing, each trader ends up with a positive exposure to  market risk,  with a beta less than one, even if initial positions are negatively correlated to market risk. Note also that traders with higher risk tolerance are willing to get relatively more exposure to the market risk through the competitive transaction.  

The cash amount (\emph{signed risk premium}) that trader $\iii$ pays to obtain post-transaction beta equal to $\lambda_i$ is 
\[
\inner{\pq_i}{\pp} = (\beta_i-\lambda_i)\inner{a_I}{Ca_I}/\delta_I,
\]
which is linearly increasing with respect to $\beta_i$. In fact, traders that reduce their beta after the competitive transaction (i.e., those with $\lambda_i<\beta_i$) pay a positive risk premium $\abs{\inner{\pq_i}{\pp}} = \inner{\pq_i}{\pp}$ to their counter-parties. On the other hand, traders that undertake market risk at the competitive transaction (i.e., those with $\beta_i<\lambda_i$) are compensated with a risk premium $\abs{\inner{\pq_i}{\pp}} = - \inner{\pq_i}{\pp}$.

Based on the formulas of equilibrium prices and allocations of \eqref{eq:p_comp_equil} and \eqref{eq:q_comp_equil}, we readily calculate and decompose the traders' utility at competitive equilibrium as
\begin{align}\label{eq: gain in equilibrium}
\U_i \pare{E_i + \inner{\pq_i}{\secu - \pp}}&= u_i+ \frac{1}{2\delta_i} \abs{C^{1/2} (\lambda_i a_I - a_i)}^2=u_i+ \frac{1}{2\delta_i} \abs{C^{1/2} \pq_i}^2 \\
\nonumber &= u_i+\underset{\text{profit/loss from random payoff}}{\underbrace{\frac{1}{2\delta_i}\inner{a_i}{Ca_i}-\lambda_i^2\frac{\inner{a_I}{Ca_I}}{2\delta_i}}}\quad\underset{\text{(signed) risk premium}}{\underbrace{-\frac{\beta_i-\lambda_i}{\delta_I}\inner{a_I}{Ca_I}}}, \quad \iii.
\end{align}
Larger trades at competitive equilibrium result in higher utility gain after the transaction. The above decomposition of utility into risk-sharing gain and risk premium allows one to analyse further the exact sources of utility for each trader, and will prove especially useful later on, when comparing competitive and noncompetitive equilibria.

\section{Traders' Best Response Problem} \label{sec:best_response}

\subsection{The setting of trader's response problem}

While it is rather reasonable to assume that pre-transaction betas are publicly known, it is problematic to impose a similar informational assumption on traders' risk profiles. We view risk tolerance as a subjective parameter, and more realistically consider it as \emph{private information} of each individual trader. In the CARA-normal market setting treated here, each trader's risk tolerance is reflected in the elasticity of the submitted demand function. In particular, from Proposition \ref{prop:comp_equilibrium} and the induced individual utility gain \eqref{eq: gain in equilibrium}, elasticities of traders' submitted demand directly affect both the allocation of market risk and the associated risk premia. Therefore, it is reasonable to inquire whether an individual trader has motive to strategically choose the elasticity of the submitted demand function. More precisely, adapting the family of linear demand functions with downward slope of the form \eqref{eq.demand}, strategically chosen elasticity is equivalent to submitting demand function
\begin{equation}\label{eq.sub demand}
Q_i^{\rtoi}(p) = - a_i - \rtoi C^{-1} p, \quad p \in \Real^K,
\end{equation}
where $\rtoi\in (0,\infty)$ is the elasticity of the submitted demand function $Q_i^{\rtoi}$; equivalently, $1/\rtoi$ is the risk aversion reflected by the submitted demand. In the extreme case where $\rtoi \rightarrow \infty$, trader $\iii$ submits extremely elastic demand, or equivalently behaves as risk neutral, while $\rtoi\rightarrow 0$ indicates extremely inelastic demand, i.e., a case where the trader does not want to undertake any risk.

The question addressed in the present section is how traders choose the elasticity of their demand function within the family of demands \eqref{eq.sub demand}, and whether this is different than their risk tolerance. In order to make headway with examining the best response function of trader $\iii$, we assume that all traders except trader $\iii$ have submitted an aggregate linear demand function of the form \eqref{eq.sub demand}, where 
$\rtoni=\sum_{j \in I \setminus \set{i}}\theta_j \in (0,\infty)$ is the aggregate elasticity of all but traders except trader $\iii$. Under this scenario, if trader $\iii$ chooses to submit the demand function \eqref{eq.sub demand} with $\rtoi \in (0,\infty)$, and recalling \eqref{eq:p_comp_equil} and \eqref{eq:q_comp_equil}, the equilibrium price and allocations will equal
\[
\pp(\rtoi;\rtoni) = - \frac{1}{\rtoi+\rtoni} C a_I, \qquad \pq_i(\rtoi;\rtoni) = \frac{\rtoi}{\rtoi+\rtoni} a_I - a_i,
\]
and hence the trader's payoff will equal
\[
E_i+\inner{\pq_i(\rtoi;\rtoni)}{\secu - \pp(\rtoi;\rtoni)}.
\] 
Since $\theta_{-i} > 0$, the limiting cases when $\theta_i = 0$ (interpreted as extreme inelasticity) and $\theta_i = \infty$ (interpreted as risk neutrality) are well defined; indeed, taking limits in the expressions above, it follows that
\begin{align*}
\pp(0;\rtoni) = - \frac{1}{\rtoni} C a_I&, \qquad \pq_i(0;\rtoni) = - a_i, \\
\pp(\infty;\rtoni) = 0&, \qquad \pq_i(\infty;\rtoni) = a_I - a_i = a_{-i}.
\end{align*}
Risk-neutral acting traders satisfy all the demand of the other traders, accepting all their market risk, without asking a risk premium (recall that we have assumed that $\expec[\secu_k] = 0$, $\forall k \in K$). On the other hand, extremely unelastic demand implies hedging all the initial position, making the post-transaction beta equal to zero and in fact delegating determination of equilibrium prices to other traders. Using the standard terminology of portfolio management, we call \emph{market-neutral} a position with zero beta.

For $\rtoni \in (0, \infty)$, and under the standing assumption of Gaussian endowments and securities made in Section \ref{sec:setup}, the \emph{response function} of trader $\iii$ is
\begin{align*}
(0,\infty) \ni \rtoi \mapsto \V_i(\rtoi;\rtoni) &\equiv  \U_i(E_i+\inner{\pq_i(\rtoi;\rtoni)}{\secu - \pp(\rtoi;\rtoni)}) \\
&u_i+\inner{\frac{\rtoi}{\rtoi+\rtoni} a_I - a_i}{\, C\pare{\frac{1}{\rtoi+\rtoni} a_I -\frac{1}{2\delta_i}\pare{\frac{\rtoi}{\rtoi+\rtoni} a_I + a_i}}},
\end{align*}
with $\rtoi$ indicating parametrisation of the trader's strategic behaviour. Since the limiting cases for $\theta_i$ are also well defined, we allow a trader to submit demand functions that declare extreme and  zero elasticity; for these cases, we have
\begin{align*}
\V_i(0;\rtoni)&=\U_i \pare{ E_i-\inner{a_i}{S} -\frac{1}{\rtoni} \inner{a_i}{C a_I} } = u_i+\frac{1}{2\delta_i}\inner{a_i}{Ca_i}-\frac{1}{\rtoni}\inner{a_i}{C a_I},\\
\V_i(\infty;\rtoni) &= \U_i(E_i+\inner{a_{-i}}{S}) = u_i-\frac{1}{2\delta_i}\inner{a_{-i}}{C(a_I + a_i)}.
\end{align*}
Summing up, given $\rtoni\in(0,\infty)$, the trader $\iii$'s best response problem is maximising the post-transaction utility by strategically chosen the submitted demand elasticity, i.e., 
\begin{equation}\label{eq.best_response_problem}
\rb(\rtoni) = \underset{\rtoi\in[0,\infty]}{\argmax}\V_i(\rtoi;\rtoni).
\end{equation}

\begin{rem} \label{rem:uninteresting_2}
When $a_I = 0$, $\V_i(\theta_i ;\rtoni) = u_i + \inner{a_i}{Ca_i} / 2\delta_i$ holds for all $\theta_i \in [0, \infty]$. In this case, the response function is flat, and any response leads to the same  equilibrium prices $\pp(\theta_i;\rtoni) = 0$ and allocation $\pq_i(\theta_i;\rtoni) = - a_i$ for trader $\iii$, irrespectively of the value of $\theta_{-i}$. These are exactly the prices and allocations one obtains at competitive equilibrium.
\end{rem}

The following result shows that, under the assumptions made in Section \ref{sec:setup} (in particular, that $a_I \neq 0$), the best response problem \eqref{eq.best_response_problem} admits a unique solution (recall that $\beta_{-i}$ denotes the difference $1-\beta_i$, which is equal to $\sum_{j \in I \setminus \set{ i}}\beta_j$). 
\begin{prop} \label{p.best_response}
	Given $\rtoni\in(0,\infty)$, the best response of trader $\iii$ exists, is unique and given as follows:
	\begin{align}\label{eq.rb}
	\rb(\rtoni)=\left\{
	\begin{array}{ll} 0, & \text{if }\beta_i\leq -1; \\
	\delta_i \rtoni(1+\beta_i) / \pare{\rtoni+\delta_i\beta_{-i}}, & \text{if } -1 < \beta_i < 1 + \rtoni / \delta_i; \\
	\infty, & \text{if }\beta_i\geq 1 + \rtoni / \delta_i. \\
	\end{array}
	\right.
	\end{align}
\end{prop}

\begin{proof}
	Fix $\rtoni\in(0,\infty)$. Making the monotone change of variable
	\[
	[0, \infty] \ni \theta_i \mapsto k_i \dfn \frac{\rtoi}{\rtoi+\rtoni}\in[0,1],
	\]
	and using a slight abuse of notation, maximizing value function $\V_i$ is equivalent to maximising 
	\begin{align} \label{eq.vkbi}
	\V_i(k_i;\rtoni) &= u_i+\inner{a_I}{C a_I}\pare{\frac{(1-k_i)k_i}{\rtoni}-\frac{k^2_i}{2\delta_i}}-\inner{a_I}{Ca_i}\frac{1-k_i}{\rtoni} \\ 
	\nonumber &= u_i+\inner{a_I}{C a_I}\pare{\frac{(1-k_i)k_i}{\rtoni} -\frac{k^2_i}{2\delta_i} - \beta_i \frac{1-k_i}{\rtoni} }.
	\end{align}
	Since $a_I \neq 0$, the above is a strictly concave quadratic function of $k_i \in[0,1]$; in particular, it has a unique maximum. When $\beta_i \leq -1$ (resp., when $\beta_i \geq 1+\rtoni/\delta_i$), it is straightforward to see that $[0,1] \ni k_i \mapsto \V_i(k_i; \theta_{-i})$ is decreasing (resp., increasing). It follows that $\rb(\rtoni)=0$ when $\beta_i \leq -1$, while $\rb(\rtoni)=\infty$ when $\beta_i \geq 1+\rtoni/\delta_i$. When $-1 < \beta_i < 1+\rtoni/\delta_i$, first-order conditions in \eqref{eq.vkbi} give that the unique maximizer of $[0,1] \ni k_i \mapsto \V_i(k_i; \theta_{-i})$ is
	\begin{equation}\label{eq.kbi}
	\kb_i(\theta_{-i})=\pare{2 + \frac{\rtoni}{\delta_i} }^{-1} \pare{1 + \beta_i}.
	\end{equation}
	It then readily follows from \eqref{eq.kbi} that the unique maximizer of $[0, \infty] \ni \theta_i \mapsto \V_i(\rtoi,\rtoni)$ is $\rb(\rtoni)=\delta_i\rtoni(1+\beta_i)/(\rtoni+\delta_i(1-\beta_i))\in(0,\infty)$. 
\end{proof}

According to Proposition \ref{p.best_response}, extreme best responses $\theta_i$ for trader $\iii$ are possible, given $\theta_{-i} \in (0, \infty)$. In fact, the best response is zero if and only if $\beta_i \leq - 1$, irrespectively of the value of $\theta_{-i}$, and the best response is infinity if and only if $\beta_i \geq 1 + \rtoni / \delta_i$. In view of this potentiality, it makes sense to understand how a trader would respond if $\theta_{-i}$ itself took an extreme value.

We start with the case $\theta_{-i} = \infty$. In this case, taking the limit as $\theta_{-i} \to \infty$ in \eqref{eq.vkbi} gives
\begin{equation}\label{eq.best_response_extreme}
\rb(\infty)=\delta_i(1+\beta_i)_+.
\end{equation}
The case $\rtoni=0$ may be treated similarly, but it is worthwhile making an observation. Note that $\rtoni=0$ means that all other traders except $\iii$ submit extremely unelastic demands. According to the solution of the best response problem, and anticipating the definition of Bayesian Nash equilibrium in Section \ref{sec.Nash}, this only makes sense when $\beta_j \leq -1$ holds for $j \in I \setminus \set{i}$. Since $\beta_i = 1 - \sum_{j \in I \setminus \set{i}} \beta_j$ and there are at least two traders, it should be that $\beta_i > 1$. In this case, taking the limit as $\theta_{-i} \to 0$ in \eqref{eq.vkbi} gives $\rb(0) = \infty$. To recapitulate: when $\theta_{-i} = \infty$ the best response is given by \eqref{eq.best_response_extreme}. The case $\theta_{-i} = 0$ is interesting only in the case $\beta_i > 1$, where we set
\[
\rb(0) = \infty, \quad \text{whenever} \quad \beta_i > 1.
\]

\smallskip


It is clear from Proposition \ref{p.best_response} that non-price-taking traders have motive to submit demand function of different elasticity than their risk tolerance. The main determinant of departure from agents' true demand is their pre-transaction beta, defined in \eqref{eq.beta}. In order to analyse the effect of strategic behaviour on the equilibrium prices and allocations, we may consider the situation where trader $\iii$ is the only one acting strategically against price-takers; all other agents submit the elasticity corresponding to their true demand functions for the transaction. In symbols, we set $\rtoni=\delta_{-i}$. This can be seen as a one-sided noncompetitive equilibrium, in the sense that only trader $\iii$ exploits knowledge on other traders' elasticity and endowments, and responds optimally. The post-transaction beta \eqref{eq.kbi} becomes $\kb_i=0$ when $\beta_i \leq -1$, $\kb_i=1$ when $\beta_i\geq 1/\lambda_i$, and $\kb_i=\lambda_i (1+\beta_i) / (1+\lambda_i)$ when $\beta_i \in (-1, 1/\lambda_i)$. In obvious terminology, we shall call the latter regime \emph{non-extreme}, while the former two will be called  \emph{extreme}.

It is completely straightforward from the closed-form expressions for $\kb_i$ that   
\[
\lambda_i<\beta_i\quad\text{if and only if}\quad\lambda_i<\kb_i<\beta_i.
\]
Taking into account the discussion following Proposition \ref{prop:comp_equilibrium}, the above fact implies that traders have motive to submit more elastic demand functions if and only if they reduce their market risk through the transaction. At the non-extreme regime, this happens when $\beta_i \in (\lambda_i, 1/\lambda_i)$, where the trader's initial position is considered \emph{relatively more exposed to market risk}. 
	
A direct outcome when acting more \textit{aggressively} by submitting more elastic demand is that the post-transaction beta entails more risk: indeed, instead of $\lambda_i \inner{a_I}{\secu}$ at competitive equilibrium, the (random part of) the portfolio after submitting demand with elasticity $\rb$ equals $\kb_i\inner{a_I}{\secu}$. In particular, the post-transaction beta of trader $\iii$ is $\kb_i$, instead of $\lambda_i$. Although the reduction of risk exposure is lower when compared to the competitive equilibrium, it comes at a better price. To wit, we readily calculate that in the whole non-extreme regime $\beta_i \in (-1, 1/\lambda_i)$ it holds that\footnote{When $\beta\in(-1,1/\lambda_i)$, the exact cash benefit from the best response strategy equals
\[
\inner{\qb_i}{\pp-\pb}=\inner{a_I}{Ca_I}\frac{\lambda_i(\beta_i-\lambda_i)^2}{\delta_I(1+\lambda_i)^2(1-\lambda_i)}.
\]
At competitive equilibrium trader $\iii$ pays $\inner{\pq_i}{\pp}=(\beta_i-\lambda_i)\inner{a_I}{Ca_I}/\delta_I$ to reduce beta exposure to $\lambda_i$, while acting strategically the trader pays $\inner{\qb_i}{\pb}=(\beta_i-\lambda_i)\inner{a_I}{Ca_I}(1-\lambda_i\beta_i)/[(\delta_I-\delta_i)(1+\lambda_i)^2]$ to reduce beta exposure to $\kb_i$. Note that $\inner{\qb_i}{\pb}<\inner{\pq_i}{\pp}$, when $\beta_i\in(\lambda_i,1/\lambda_i)$.}
\[
\inner{\qb_i}{\pb}<\inner{\qb_i}{\pp},
\]
which means that the gain of the strategic behaviour comes from the lower premium that is paid.

\begin{rem}\label{rem:comparative statics}
Under the very special case $\beta_i=\lambda_i$, one obtains $\rb(\rtoni)=\delta_i$, i.e., $\kb_i=\lambda_i$. In view of \eqref{eq:q_comp_equil}, the latter condition implies $\pq_i =0$ and hence trader $\iii$ does not participate in the sharing of risk; this is also the case in competitive equilibrium.
\end{rem}

\section{Noncompetitive Risk-Sharing Equilibrium}\label{sec.Nash}

\subsection{Nash equilibrium}

With the best response problem \eqref{eq.best_response_problem} in mind, and assuming that all traders act strategically, we now address noncompetitive Bayesian Nash equilibrium. More precisely, in a fashion similar to the demand-submission game of \cite{Kyl89}, traders submit linear demand schedules of the form \eqref{eq.sub demand}, where $(\rtoi)_{\iii} \in [0,\infty]^{I}$ and $\theta_I = \sum_{\iii}\rtoi>0 $ are the corresponding individual and aggregate  submitted demand elasticity. The market equilibrates at the pairs of prices and allocations at which the submitted demands sum up to zero. According to Proposition \ref{prop:comp_equilibrium}, as well as relations \eqref{eq:p_comp_equil} and \eqref{eq:q_comp_equil},  
for every submitted demands with elasticities  $(\rtoi)_{\iii} \in [0,\infty]^{I}$, 
the prices and allocations that clear out the market are given by $\pp((\rtoi)_{\iii})=- (1/\theta_I) C a_I$, as well as $\pq_j ((\rtoi)_{\iii}) = (\rto_j/\theta_I) a_I - a_j,$ for each $\jii$. In other words, traders' strategies are parametrised by their submitted elasticity within the family of linear demands \eqref{eq.sub demand}, according to the best response \eqref{p.best_response}, and noncompetitive equilibria are fixed points of these responses. 

\begin{defn}\label{d.Nash}
A vector $(\rgi)_{\iii}\in[0,\infty]^{I}$, with $\rg_I\dfn\sum_{\iii}\rgi>0$, is called \textbf{Nash equilibrium} or \textbf{noncompetitive equilibrium} if, for each $\iii$, 
\[
\V_i(\rgi;\rgni)\geq\V_i(\rtoi;\rgni),\quad\forall\rtoi\in[0,\infty].
\]
By a slight abuse of terminology, we also call a \textbf{Nash price-allocation equilibrium} the corresponding pair $(\pg,(\qgi)_{\iii})\in\Real^K\times\Real^{K \times I}$ given by  \begin{equation}\label{eq:nash_equil}
	\pg = - \frac{1}{\rg_I} C a_I\quad\text{and}\quad \qgi = \frac{\rgi}{\rg_I} a_I - a_i, \quad \iii.
	\end{equation}
	where we set $\rgi / \rg_I = 1$ whenever $\rgi = \infty$, by convention.
\end{defn}

From the discussion of Section \ref{sec:best_response}, and particularly given \eqref{eq.rb} and \eqref{eq.best_response_extreme}, the possibility of noncompetitive equilibrium where some traders behave as being risk neutral (i.e., $\rg_i = \infty$ for some $\iii$) arises. We shall call such Nash equilibria where $\rg_I = \infty$ \emph{extreme}, and any other case where the total elasticity $\rg_I$ belongs to $(0,\infty)$ will be called \emph{non-extreme}.

\begin{rem} \label{rem:uninteresting_3}
When $a_I = 0$, it follows from Remark \ref{rem:uninteresting_2} that any vector $(\theta_i)_{\iii} \in \Real_+^{I}$ is a Nash equilibrium, resulting always in the same Nash price-allocation with $\pg = 0$ and $\qg_i = - a_i$ for all $\iii$. Therefore, prices and allocations at competitive and Nash equilibrium coincide. In the sequel, we continue the analysis by excluding this trivial case $a_I = 0$.
\end{rem}

\begin{rem}\label{rem:literature_comp}
Having defined our notion of noncompetitive equilibrium, we highlight its differences with the thin market models studied in \cite{RosWer15, MalRos14}. As pointed out in the introductory section, the price impact in these papers equals the slope of the aggregate demand submitted by other traders. Traders respond to---or equivalently, trade against---the price impact of their counter-parties forming a slope-game; see \cite[Lemma 1]{RosWer15} and \cite[Proposition 1]{MalRos14}. Our model keeps the form of equilibrium similar to the competitive one, as the family of demands are linear and of the form \eqref{eq.sub demand}; furthermore, although we parametrise traders' strategies to the single control variable that is elasticity, the key element is that responses, and hence equilibrium conditions, take into account the whole demand function of other traders. 
\end{rem}

Our main goal in the sequel is to study existence and uniqueness of the aforementioned linear Bayesian Nash equilibrium, and compare it with the competitive one. Departure from competitive market structure reduces the aggregate transaction utility gain. Indeed, it can be easily checked (see, for example, \cite[Corollary 5.7]{AnthZit10}) that the allocation $(\pq_i)_{\iii}$ of \eqref{eq:q_comp_equil} maximises the sum of traders' monetary utilities over all possible market-clearing allocations. As utilities given by \eqref{eq:utility_functional} are monetary, we can measure the risk-sharing inefficiency of any noncompetitive equilibrium as the difference between  \textit{aggregate} utility at Nash and competitive equilibrium.

We shall verify in the sequel that risk sharing in the noncompetitive equilibrium is, except in trivial cases, socially inefficient. However, it is not necessarily true that each individual trader's utility is reduced; in fact, it is reasonable to ask which (if any) traders prefer Nash risk sharing in such a thin market, as opposed to the corresponding market that equilibrates in competitive manner. For this, we compare the individual utility gains at two equilibria, that is, the difference 
\begin{equation}\label{eq.comparison_utilities}
\DU_i\equiv \underset{\text{utility at Nash equilibrium}}{\underbrace{\U_i \pare{E_i + \inner{\qg_i}{\secu - \pg}}}} - \underset{\text{utility at competitive equilibrium}}{\underbrace{\U_i \pare{E_i + \inner{\pq_i}{\secu - \pp}}},} \quad \text{ for each }\iii,
\end{equation} 
and ask when this is positive. Given this notation, and as discussed above, the inefficiency  of the noncompetitive risk-sharing is defined as the sum $\sum_{\iii}\DU_i$.

\subsection{Equilibrium with at most one trader's beta being greater than one}\label{subsec:More_than_two}

Under the condition\footnote{We conjecture that Theorem \ref{thm:ex_and_un} is true in all cases, although we do not have a rigorous proof of this claim.} that at most one of the traders have initial beta higher than one, that is
\begin{equation} \label{eq:beta_more_one_up_to_one}
\# \set{\iii \such \beta_i > 1} \in \{ 0,1 \},
\end{equation}
the next result states that there exists a unique \emph{linear} noncompetitive equilibrium.

\begin{thm} \label{thm:ex_and_un}
Under \eqref{eq:beta_more_one_up_to_one}, there exists a unique Nash equilibrium as in Definition \ref{d.Nash}. 
\end{thm}

According to \eqref{eq.rb}, traders behave as being risk neutral when their initial exposure to market risk is sufficiently higher than one. As we will show in Proposition \ref{pr.extreme} below, this behaviour pertains at equilibrium, making it an extreme one, if and only if the following condition holds:
\begin{equation}\label{eq.condition_for_extreme_bis}
\sum_{\iii}\delta_i(1+\beta_i)_+ \leq 2 \max_{\iii} (\delta_i \beta_i ).
\end{equation}  
When \eqref{eq.condition_for_extreme_bis} fails, the (unique) Nash equilibrium is non-extreme; in this case, and in view of \eqref{eq.rb}, the following coupled system of equations 
\begin{equation}\label{eq.Nash_couples} 
\pare{2 + \frac{\rg_I - \rgi}{\delta_i} } \frac{\rgi}{\rg_I} =  1 + \beta_i, \quad \forall \iii\text{ with }\beta_i>-1,
\end{equation}
should hold, where it is $\rg_I$ which couples the equations. According to  \eqref{eq.rb}, any trader $\iii$ with $\beta_i \leq -1$ optimally submits demand function with zero elasticity, inducing a market-neutral post-transaction position, where recall that a position is called market-neutral when it has zero induced beta. Theorem \ref{thm:ex_and_un} states, in particular, that the system \eqref{eq.Nash_couples} admits a unique solution for an arbitrary number of traders when \eqref{eq.condition_for_extreme_bis} fails. This fact is proved in Appendix \ref{sec:appe}, and it is important to note that the proof is \emph{constructive}, and hence can be used to numerically calculate the equilibrium quantities when the number of traders is more than two; the case of two traders admits in fact a closed-form solution and is extensively studied in \S \ref{subsec.two_agents} later on.

\subsection{Risk-neutral behaved trader(s)} \label{subsec:extreme_case}

Having established existence and uniqueness of Nash equilibrium in Theorem \ref{thm:ex_and_un}, we now show that the condition \eqref{eq.condition_for_extreme_bis} necessarily leads to an extreme noncompetitive equilibrium. We start with an alternative characterisation of\eqref{eq.condition_for_extreme_bis}.
\begin{lem} \label{lem:condition_for_extreme}
Condition \eqref{eq.condition_for_extreme_bis} is equivalent to
\begin{equation} \label{eq.condition_for_extreme}
\beta_k \geq 1 + \frac{1}{\delta_k} \sum_{i \in I \setminus \set{k}} \delta_i (1+\beta_i)_+, \quad \text{for some} \quad k \in I.
\end{equation}  
Furthermore, \eqref{eq.condition_for_extreme} can hold for at most one trader $k \in I$.
\end{lem}

\begin{proof}
Start by assuming that \eqref{eq.condition_for_extreme} holds, and rewrite it as $\delta_k \beta_k \geq \delta_k + \sum_{i \in I \setminus \set{k}}\delta_i(1+\beta_i)_+$. Since $\beta_k > 1$, which implies that $1 + \beta_k = (1 + \beta_k)_+$, adding $\delta_k \beta_k$ on both sides of the previous inequality and simplifying, we obtain $2 \delta_k \beta_k \geq \sum_{i \in I} \delta_i (1+\beta_i)_+$, from which \eqref{eq.condition_for_extreme_bis} follows. Conversely, \eqref{eq.condition_for_extreme_bis} holds if and only if $2 \delta_k \beta_k  \geq \sum_{\iii}\delta_i(1+\beta_i)_+$ holds for some $k \in I$. In this case, $\beta_k \geq 0 > -1$, and subtracting $\delta_k (1 + \beta_k) = \delta_k (1 + \beta_k)_+$ we obtain $\delta_k (\beta_k - 1)  \geq \sum_{i \in I \setminus \set{k}} \delta_i (1+\beta_i)_+$, which is  \eqref{eq.condition_for_extreme}.
	
Assume now that \eqref{eq.condition_for_extreme} held for two traders, say trader $k \in I$ and $l \in I$ with $k \neq l$. Then, 
\[
\delta_k (\beta_k - 1) \geq \sum_{i \in I \setminus  \set{k}} \delta_i (1+\beta_i)_+ \geq \delta_\ell ( 1 + \beta_\ell) \quad\text{and}\quad \delta_l (\beta_l - 1) \geq \sum_{i \in I \setminus  \set{l}} \delta_i (1+\beta_i)_+ \geq \delta_k ( 1 + \beta_k).
\]
Adding up these inequalities we obtain $- 2 (\delta_k+\delta_l) \geq 0$, which contradicts the fact that $\delta_k > 0$ and $\delta_l > 0$. We conclude that \eqref{eq.condition_for_extreme} can hold for at most one trader.
\end{proof}

The next result gives a complete characterisation of extreme noncompetitive equilibrium; in particular, it shows that at most one trader---and, in fact, exactly the trader $k \in I$ for which \eqref{eq.condition_for_extreme} holds---may behave as risk-neutral in noncompetitive equilibrium. Note that we do \emph{not} assume \eqref{eq:beta_more_one_up_to_one} for Proposition \ref{pr.extreme}, as it was also not needed for Lemma \ref{lem:condition_for_extreme}

\begin{prop} \label{pr.extreme}
An extreme noncompetitive equilibrium (i.e., with $\rg_I = \infty$) exists if and only if \eqref{eq.condition_for_extreme_bis}, or equivalently \eqref{eq.condition_for_extreme}, is true. In this case, we have $\rg_k = \infty$ for the unique trader $k \in I$ such that \eqref{eq.condition_for_extreme} holds, and $\rg_i = \delta_i (1 + \beta_i)_+$ for $i \in I \setminus \set{k}$. In particular, the previous is the unique extreme noncompetitive equilibrium under the validity of \eqref{eq.condition_for_extreme_bis}.
\end{prop}

\begin{proof}
First, assume that a Nash equilibrium with $\rg_I = \infty$ exists. Since $\# I < \infty$, there exists $k \in I$ with $\rg_k = \infty$. According to \eqref{eq.best_response_extreme}, for any trader $i \in I \setminus \{k\}$, it holds that $\rg_i = \delta_i (1+\beta_i)_+$. Therefore, for $\rg_k = \infty$ to be the best response for trader $k \in I$, \eqref{eq.rb} gives $\beta_k \geq 1 + \pare{1 / \delta_k} \sum_{i \in I \setminus \set{k}} \delta_i (1+\beta_i)_+$. It follows that \eqref{eq.condition_for_extreme} is a necessary condition for existence of an extreme noncompetitive equilibrium.
	
Conversely, if \eqref{eq.condition_for_extreme} holds, and defining $\rg_k = \infty$ and $\rg_i = \delta_i (1 + \beta_i)_+$ for $i \in I \setminus \set{k}$, it is immediate from \eqref{eq.rb} and \eqref{eq.best_response_extreme} to check that the previous is indeed a Nash equilibrium.
\end{proof}

We proceed with some discussion, where we assume that \eqref{eq.condition_for_extreme} holds. In view of Proposition \ref{pr.extreme} and the relations in \eqref{eq:nash_equil}, at the extreme equilibrium trader $k \in I$ undertakes all market risk, since $\qg_k = a_I - a_k$, and the rest of the traders exchange all their market risk (i.e., $\qg_i = -a_i$, for each $i \in I \setminus \{ k \}$) at zero cost, since pricing is done in a risk-neutral way ($\pg = 0$). In particular, the post Nash-transaction beta of trader $k \in I$ reduces to one, and all other traders become market-neutral.

While this transaction is not socially optimal, participating traders increase their utilities; otherwise, equilibrium would not form. Straightforward calculations give the individual utility gains at the extreme equilibrium: $\U_k \pare{E_k + \inner{\qg_k}{\secu - \pg}} = u_k + \left(\inner{a_k}{C a_k}-\inner{a_I}{Ca_I} \right) / 2\delta_k$ and $\U_i \pare{E_i + \inner{\qg_i}{\secu - \pg}}=u_i+\inner{a_i}{Ca_i}/2\delta_i$, for each $i\in I \setminus \{ k \}$. In particular, the difference of utility gains in \eqref{eq.comparison_utilities} between the extreme Nash equilibrium and the competitive one equal
\begin{equation}\label{eq.extreme individual gains}
\DU_k = \frac{\inner{a_I}{C a_I}}{2\delta_k} \left[ \lambda_k (2\beta_k-\lambda_k)-1 \right], \quad \text{and} \quad \DU_i = \frac{\inner{a_I}{Ca_I}}{2\delta_i}\lambda_i(2\beta_i-\lambda_i),\  \forall i\in I \setminus \{k\}.
\end{equation}
It follows by straightforward algebra that
\[
\text{Risk-sharing inefficiency} \dfn \sum_{\iii}\DU_i = - \frac{\inner{a_I}{C a_I}}{2\delta_I} \frac{1-\lambda_{k}}{\lambda_k}.
\]
As expected, there is a reduction of the total utility gain when traders behave strategically regarding the elasticity of their submitted demand functions. However, utility gains may be higher in the noncompetitive equilibrium for \emph{individual} traders. From \eqref{eq.extreme individual gains}, we conclude that, in extreme noncompetitive equilibrium, traders that benefit from the market's thinness are the ones with sufficiently high initial exposure to market risk: for trader $k \in I$, when $\beta_k > (1 + \lambda_k^2) / 2 \lambda_k$ and for traders $i\in I\setminus\{k\}$ when $\beta_i > 2\lambda_i$.\footnote{As easy examples show, condition \eqref{eq.condition_for_extreme} does not necessarily imply $\beta_k > (1+\lambda_k^2)/2\lambda_k$. In the special two-trader case $I=\{0,1\}$ with $k = 0$, condition \eqref{eq.condition_for_extreme} is equivalent to $\beta_0 > 2 - \lambda_0$, which always implies $\beta_0 > (1+\lambda_0^2)/2\lambda_0$ when $\lambda_0 > 1/3$. Still in the same two-trader case with $k=0$, condition \eqref{eq.condition_for_extreme} implies that $\beta_1 < 2 \lambda_1$: in the bilateral extreme equilibrium, only the trader that acts as risk neutral could benefit from the market's thinness.} 

The above quantitative discussion has the following qualitative attributes. Under condition \eqref{eq.condition_for_extreme}, in the noncompetitive extreme equilibrium trader $k \in I$ reduces market-risk exposure to one but pays zero premium to other traders. If the market's equilibrium was competitive, trader $k \in I$ would decrease the post-transaction beta even more, to $\lambda_0$ instead to one, but the premium would be strictly positive according to the decomposition \eqref{eq: gain in equilibrium}. The benefit of zero risk premium prevails the lower reduction of risk if $\beta_k$ is sufficiently large. On the other hand, the rest of the traders sell all their market-risk exposure at zero premium. For those traders with low initial beta (more precisely, $\beta_i<\lambda_i/2$), the noncompetitive equilibrium leaves them worse off than the competitive one. This stems from the fact that in competitive equilibrium  traders with low initial beta obtain premium from traders who are overexposed to market risk, something that does not occur in the noncompetitive extreme equilibrium. However, for traders with $\beta_i\geq\lambda_i/2$, the noncompetitive equilibrium is preferable since they also benefit from the zero risk premium.

\emph{To recapitulate: traders that obtain more utility from the extreme noncompetitive equilibrium are the ones with sufficiently high initial exposure to market risk.}

\section{Bilateral Strategic Risk Sharing}\label{sec.bilateral}

\subsection{The case of essentially two strategic traders}\label{subsec.two_agents}

As pointed out in the introductory section, the two-trader case is of special interest since the majority of the OTC transactions consists of only two institutions, or one institution and a client.

Since traders with pre-transaction beta less or equal to $-1$ always sell all their risk at equilibrium, a risk-sharing game is essentially between two traders if exactly two of them (for concreteness's sake, traders $0 \in I$ and $1 \in I$) have pre-transaction beta larger than $-1$. Then, traders 0 and 1 are the only ones to submit demands with non-zero elasticity. In view of the general analysis of \S \ref{subsec:extreme_case}, we shall only treat the case of non-extreme equilibrium, i.e., when \eqref{eq.condition_for_extreme_bis} fails. Straightforward algebra yields that, in the present case, failure of \eqref{eq.condition_for_extreme_bis} is equivalent to the following simplified inequality
\begin{equation} \label{eq.nash_two_agents}
|\lambda_0\beta_0-\lambda_1\beta_1|<\lambda_0+\lambda_1.
\end{equation}
If $I = \{0,1\}$, and recalling that $\beta_0+\beta_1=\lambda_0+\lambda_1=1$ in this case, inequality \eqref{eq.nash_two_agents} is equivalent to $-\lambda_i<\beta_i<2-\lambda_i$ for both $i \in \{0,1\}$.

\begin{prop} \label{pr.nash_two_agents}
	Assume that $\beta_0 > -1$, $\beta_1>-1$, $\beta_i \leq -1$ for $\iii \setminus \{0,1\}$, and impose \eqref{eq.nash_two_agents}. Then a noncompetitive equilibrium is unique, satisfies $\rg_i = 0$ for $\iii \setminus \{0,1\}$, as well as
	\begin{equation}\label{eq.rb_two}
	\rg_0 = \delta_0\frac{2\lambda_1(\beta_0+\beta_1)}{(\lambda_0+\lambda_1)+(\lambda_1\beta_1-\lambda_0\beta_0)}, \qquad \rg_1 = \delta_1\frac{2\lambda_0(\beta_0+\beta_1)}{(\lambda_0+\lambda_1)+(\lambda_0\beta_0-\lambda_1\beta_1)}.
	\end{equation}
\end{prop}

\begin{proof}
As already mentioned, Proposition \ref{p.best_response} implies that the best response for each trader $\iii$ with $\beta_i \leq -1$ is zero; for traders 0 and 1, $\rg_0$ and $\rg_1$ should satisfy \eqref{eq.Nash_couples}. In this case of essentially two traders, the system takes the form of the following two equations
\begin{equation}\label{eq.two traders game}
(2\delta_0+\rg_1)\rg_0=\delta_0(1+\beta_0)(\rg_0+\rg_1) \quad\text{ and }\quad (2\delta_1+\rg_0)\rg_1=\delta_1(1+\beta_1)
	(\rg_0+\rg_1).
\end{equation}
Subtracting the first equation from the second and dividing by $\rg_I=\rg_0+\rg_1$ gives
\begin{equation}\label{eq.two traders game1}
2(\delta_1\kg_1-\delta_0\kg_0)=\delta_1(1+\beta_1)-\delta_0(1+\beta_0),
\end{equation}
where $\kg_i\equiv\rg_i/(\rg_0+\rg_1)$ for $i \in \set{0,1}$. Since $\kg_1 = 1-\kg_0$, \eqref{eq.two traders game1} is a simple linear equation of $\kg_0$ whose unique solution is 
\begin{equation}\label{eq.two traders game2}
\kg_0=\frac{1}{2}+\frac{\lambda_0\beta_0-\lambda_1\beta_1}{2(\lambda_0+\lambda_1)}.
\end{equation}
The first equation in \eqref{eq.two traders game} can be written as $(2\delta_0+\rg_1)\kg_0=(1+\beta_0)\delta_0$, which together with \eqref{eq.two traders game2} implies that $\rg_1$ should be given as in \eqref{eq.rb_two}. A symmetric argument shows that $\rg_0$ should also be given as in \eqref{eq.rb_two}. Finally, note that assumption \eqref{eq.nash_two_agents} and the imposed condition $\beta_i \leq -1$, for each $i \in I \setminus \set{0,1}$ guarantee that both $\rg_0$ and $\rg_1$ are strictly positive and finite. 
\end{proof}

At the above noncompetitive equilibrium, prices are given by $\pg = -C a_I/(\rg_0+\rg_1)$, while the allocation is $\qg_i=a_I\rg_i/(\rg_0+\rg_1)-a_i$ for each $\iii$, i.e.~only trader 0 and 1 are left with market risk after the transaction. 

\begin{rem}
As can be readily checked, a combination of Proposition \ref{pr.extreme}, Theorem \ref{thm:ex_and_un} and Proposition \ref{pr.nash_two_agents} completely covers all possible configurations for trades including up to three players. On the other hand, one may find a configuration of four traders that is not covered by the results; for example, with $I = \set{0,1,2,3}$ and $\delta_i = 1$ for all $i \in I$, let $\beta_0 = \beta_1=2$, $\beta_2 = 0$, $\beta_3 = -3$.
\end{rem}

For the rest of this section we focus our analysis and discussion on bilateral transactions, where we assume that $I=\{0,1\}$. For the reader's convenience, we note the following result stemming immediately from Proposition \ref{pr.nash_two_agents}.

\begin{cor}\label{cor:bilateral}
When $I=\{0,1\}$ and under inequality \eqref{eq.nash_two_agents}, there is a unique linear noncompetitive equilibrium given by
\[
(\rg_0,\rg_1)= \pare{\delta_0\frac{2\lambda_1}{\lambda_1+\beta_{1}},\delta_1\frac{2\lambda_0}{\lambda_0+\beta_{0}}}.
\]
The corresponding price-allocation equilibrium is given by 
\begin{equation}\label{eq.p_Nash_two_agents}
\pg =  - \frac{ \delta_I (\lambda_0+\beta_{0}) (\lambda_1+\beta_{1})}{4 \delta_0 \delta_1} C a_I = \frac{(\lambda_0+\beta_{0})(\lambda_1+\beta_{1})}{4\lambda_0\lambda_1} \pp,
\end{equation}
and
\[
\qg_i=\frac{\lambda_i+\beta_i}{2}a_I-a_i=\frac{\pq_i}{2}+\frac{\beta_ia_I-a_i}{2}, \quad i \in \set {0,1}.
\] 
\end{cor}

\begin{rem}\label{rem:competitive=Nash}
The only case where the allocation at noncompetitive equilibrium coincides with the competitive one is when $\beta_0 = \lambda_0$, which necessarily implies that $\beta_1=\lambda_1$ also holds. This equality, however, means that the competitive equilibrium is a trivial no-transaction equilibrium, since \eqref{eq:q_comp_equil} gives $\qg_0 = 0 = \qg_1$.
\end{rem}

As expected from the analysis of Section \ref{sec:best_response}, relatively higher initial exposure to market risk implies more higher submitted elasticity at the noncompetitive equilibrium: for each $i \in \set{0,1}$,
\[
\delta_i<\rg_i\quad\Leftrightarrow\quad\lambda_i<\beta_i\quad\Leftrightarrow\quad \lambda_{-i}>\beta_{-i}.
\]
In particular, the trader who reduces (resp., increases) exposure to market risk through the transaction submits a demand function with higher (resp., less) elasticity than the one that corresponds to that trader's risk tolerance.

The above analysis implies that the trader with higher initial exposure to market risk is willing to retain some of this risk in exchange of a lower risk premium. Correspondingly, the trader who undertakes further market risk through the transaction tends to behave in more risk averse way, hesitating to undertake more risk at the same risk premium. The direct outcome is that the volume of risk sharing is lower than the one obtained at competitive equilibrium, which leads to \emph{risk-sharing inefficiency}. In fact, simple calculations yield that the Nash post-transaction beta of trader $i \in I$ changes from $\beta_i$ to $(\lambda_i+\beta_i)/2$, instead of a competitive---and socially optimal---post-transaction beta of $\lambda_i$. In other words, for both traders the noncompetitive equilibrium transaction makes their betas exactly equal to the middle point between the initial and the socially optimal ones.


\begin{rem}\label{rem:competitive price=Nash price}
From \eqref{eq.p_Nash_two_agents}, we can easily see that $\pg=\pp$ holds if and only if $\lambda_0=\beta_0$ or $\lambda_0=(2-\beta_0)/3$. While the former case is the trivial one (with zero volume at any equilibrium), the latter gives a special non-trivial case where prices remain unaffected by the traders' strategic behaviour. In this case, the Nash post-transaction  beta is $(\lambda_i+\beta_i)/2=(1+\beta_i)/3 = \lambda_{-i}$ for both $i \in \set{0,1}$. 
\end{rem}

Similar to the decomposition of utility gains at competitive equilibrium in \eqref{eq: gain in equilibrium}, we decompose the corresponding utility gains at noncompetitive equilibrium for $i \in \set{0,1}$ as
\begin{equation}\label{eq: gain in equilibrium at nash}
\U_i \pare{E_i + \inner{\qg_i}{\secu - \pg}} =u_i+\underset{\text{profit/loss from random payoff}}{\underbrace{\frac{1}{2\delta_i}\inner{a_i}{Ca_i}-\left(\frac{\lambda_i+\beta_i}{2}\right)^2\frac{\inner{a_I}{Ca_I}}{2\delta_i}}}\quad\underset{\text{(signed) risk premium}}{\underbrace{-\frac{\beta_i-\lambda_i}{\delta_I}\inner{a_I}{Ca_I}L}}, 
\end{equation}
where $L \equiv (\beta_0+\lambda_0) (\beta_1+\lambda_1)/8\lambda_0\lambda_1$. The decompositions \eqref{eq: gain in equilibrium} and \eqref{eq: gain in equilibrium at nash} give an expression for the utility difference between the two equilibria $\DU_i$ defined in \eqref{eq.comparison_utilities}; to wit,
\begin{equation}\label{eq:utilities comparison}
\DU_i= \frac{\inner{a_I}{Ca_I}}{2\delta_i}\left[\lambda_i^2-\left(\frac{\lambda_i+\beta_i}{2}\right)^2\right]+ \frac{\beta_i-\lambda_i}{\delta_I}\inner{a_I}{Ca_I}(1-L),\quad i \in \set{0,1}.
\end{equation}

As was the case in extreme equilibrium discussed in \S \ref{subsec:extreme_case}, the difference of utility gains stems from two sources: the gain from sharing the random (risky) payoffs and the risk premium paid or received. Let assume without loss of generality that $\beta_0 < \lambda_0$ (or equivalently, that $\beta_1 > \lambda_1$), i.e., that trader 0 undertakes more market risk after the (competitive or not) transaction. Since noncompetitive risk-sharing beta reaches only halfway compared to competitive risk-sharing, there is less risk undertaken by trader 0. The risk premium received for undertaking market risk is higher than the one in competitive equilibrium if and only if $L > 1$, which holds in particular when $\lambda_0$ is close to one. 
When $\lambda_0$ is not close to one, the risk premium is lower and could absorb all the gain from the lower undertaken market risk. Hence, for traders who undertake market risk at the transaction and have risk preferences close to risk neutrality, the noncompetitive equilibrium is more beneficial.

On the other hand, trader 1 is selling market risk, with lower reduction of Nash post-transaction beta (a fact that decreases utility), while the premium is lower at Nash equilibrium if and only if $L < 1$. The difference $1-L$ is negative for $\lambda_1$ close to zero, and the total difference \eqref{eq:utilities comparison} for $i=1$ remains negative when $\beta_1<1$ for every value of $\lambda_1$. For fixed $\lambda_1$, $L$ is decreasing in $\beta_1$ (when $\beta_1>1-\lambda_1$) and the total difference \eqref{eq:utilities comparison} for $i=1$ is positive when $\beta_1$ is close to its upper bound $2-\lambda_1$.

Finally, it should be pointed out that when the risk preferences of trader 0 (i.e., the buyer of market risk) are close to risk neutrality (that is, when $\lambda_0$ is close to 1), the noncompetitive equilibrium is always better than the competitive one if and only if $|\beta_{0}|<1$ or, equivalently, when $0 < \beta_{1} < 2$. In particular, \eqref{eq:utilities comparison} and the discussion of extreme equilibrium in \S \ref{subsec:extreme_case} imply that 
\[
\lim_{\delta_0 \to \infty} \DU_0 = 
\left \{
\begin{array}{ll} 
\inner{a_I}{Ca_I}(1+\beta_0)(1-\beta_0)^2 / 8 \delta_{1}, & \text{if }\beta_0\in(-1,1); \\
0, & \text{otherwise}.
\end{array}
\right.
\] 
Therefore, within non-extreme Nash equilibrium, traders that obtain more utility in the noncompetitive equilibrium are the ones with risk preferences close to risk neutrality.

In overall, we conclude that \emph{in two-trader transactions, traders that benefit with more utility from the noncompetitive equilibrium are the ones with sufficiently high initial exposure to market risk, and traders with sufficiently high risk tolerance.}

\subsection{The effect of incompleteness in thin markets}\label{subsec:incompleteness}

As emphasised above, our model allows the market to be incomplete, in that the tradeable securities do not necessarily belong to the span of the traders' endowments. When traders' endowments are not securitised, risk-sharing through competitive trading of other securities is sub-optimal. The goal of this section is to examine the effect of market's incompleteness, both on aggregate and individual levels, when the risk-sharing is \emph{noncompetitive}. For this goal, we consider the indicative two-trader game, $I=\{0,1\}$. 

In order to examine the effect of market's incompleteness we compare two market settings: an incomplete one, and one where $\secu = E$. To highlight the effect of incompleteness, we assume that besides (lack of) completeness, the rest of the parameters are the same; in particular, risk aversions remain the same, and projected and actual betas are equal. We take into account the individual utility gains \eqref{eq: gain in equilibrium}, \eqref{eq: gain in equilibrium at nash} and utility difference \eqref{eq:utilities comparison}. For quantities pertaining to the complete market we use notation with superscript ``$o$'', that is, $(q_i^o,p^o)$ are the noncompetitive equilibrium allocations and price and $\hat{q}^o_i$ the allocation under competitive equilibrium. Straightforward calculations give the following decomposition of utility gains, in terms of gains in competitive equilibrium and the effect of market noncompetitiveness:
\begin{align*}
\text{Utility gain in incomplete setting}= \U_i \pare{E_i + \inner{\qg_i}{\secu - \pg}} - u_i &=\underset{\text{gain in competitive equilibrium}}{\underbrace{\frac{1}{2\delta_i} \abs{C^{1/2} \pq_i}^2 }}+\DU_i\\
\text{Utility gain in complete setting}=\U_i \pare{E_i + \inner{q_i^o}{E - p^o}} - u_i &=\underset{\text{gain in competitive equilibrium}}{\underbrace{\frac{1}{2\delta_i} \abs{\cov^{1/2}(E,E) \hat{q}^o_i}^2 }}+\DU^o_i.
\end{align*}
Based on the above, we notice the following: The first term represents the gains of the risk-sharing if the markets were competitive. In particular, we have that (see also Proposition 2.7 in \cite{Anth17})
\[
\abs{C^{1/2} \pq_i}^2=\cov(\secu,\lambda_iE_I-E_i)C^{-1}\cov(\secu,\lambda_iE_I-E_i)\leq \var(\lambda_iE_I-E_i)=\abs{\cov^{1/2}(E,E) q^o_i}^2,
\]
where equality holds if, and only if, $\secu$ belongs in the span of $E$. The above inequality means that, under a competitive market setting, each trader loses utility due to market's incompleteness. 

The effect of market's incompleteness on the noncompetitive transaction, after accounting for the differences in the competitive environment, is captured by the difference $\DU^o_i-\DU_i$. In view of \eqref{eq:utilities comparison}, we have 
\begin{equation}
\DU_i= \frac{\inner{a_I}{Ca_I}}{2\delta_i}\left[\lambda_i^2-\left(\frac{\lambda_i+\beta_i}{2}\right)^2 + 2\lambda_i(\beta_i-\lambda_i)(1-L)  \right],\quad i \in \set{0,1}.
\end{equation} 
Keeping the parameters $\beta_i,\lambda_i$ equal for the complete and incomplete market settings, the only difference stems from the term $\inner{a_I}{Ca_I}$. In the incomplete market setting this term equals $\cov(\secu,E_I)C^{-1}\cov(\secu,E_I)$,
while in the complete market setting it equals $\var(E_I)$. Since 
\begin{equation}\label{eq: variances}
\cov(\secu,E_I)C^{-1}\cov(\secu,E_I)\leq\var(E_I),
\end{equation}
market incompleteness decreases (resp.,~increases) the  utility gain (resp.,~loss) that is caused by the market's noncompetitiveness. In other words, traders that benefit from the noncompetitive market setting (i.e., those with high risk tolerance and/or high exposure to market risk), have their utility gains reduced by the fact that endowments are not securitised. More precisely, we have seen that traders with relatively high exposure to market risk behave as risk neutral in order to reduce their exposure to one without paying risk premium. When the market is complete the reduction of the risk is more effective, since the traders sell part of their endowments and not a security that is simply positively correlated with their endowments, as in the incomplete setting. Recall also that the utility gain of the traders with relatively lower risk aversion under noncompetitive setting stems from the lower (resp., higher) risk premium that they pay (resp., receive). From \eqref{eq: gain in equilibrium at nash}, we get that the risk premium is always higher in the complete market setting (see also \eqref{eq: variances}) and hence the aforementioned increase (resp., decrease) of risk premium is also higher in the competitive setting.

We may conclude that, although market's incompleteness reduces the aggregate efficiency of risk-sharing, it also reduces the differences of utility gains/losses among traders.


\appendix

\section{Proof of Theorem \ref{thm:ex_and_un}}\label{sec:appe}

Let us start with the case where both conditions \eqref{eq.condition_for_extreme} and \eqref{eq:beta_more_one_up_to_one} hold. Proposition \ref{pr.extreme}, together under the validity of \eqref{eq.condition_for_extreme_bis} which is equivalently to \eqref{eq.condition_for_extreme}, shows that the exists an extreme linear Nash equilibrium. This is, in fact, unique over \emph{extreme} linear Nash equilibria; indeed, note that \eqref{eq.condition_for_extreme} immediately implies $\beta_k > 1$. Let us assume that trader $k \in I$ is the \emph{only} trader with beta greater than one, i.e., that $\beta_i \leq 1$ for all $i \in I \setminus \set{k}$. Let $(\rg_i)_{\iii}$ be \emph{any} linear noncompetitive equilibrium in terms of Definition \ref{d.Nash}. According to \eqref{eq.rb}, $\beta_i \leq 1$ implies that  $\rg_i \leq \delta_i (1+\beta_i)_+$, for all $i \in I \setminus \set{k}$. But then,
\[
\beta_k \geq 1 + \frac{1}{\delta_k} \sum_{i \in I \setminus \set{k}} \delta_i (1+\beta_i)_+ \geq 1 + \frac{\rg_{-k}}{\delta_k}.
\]
By \eqref{eq.rb} again, it follows that $\rg_k = \infty$, and applying \eqref{eq.rb} once again, we have $\rg_i = \delta_i (1+\beta_i)_+$, for all $i \in I \setminus \set{k}$, which establishes uniqueness of the extreme Nash equilibrium over \emph{all} possible linear Nash equilibria of Definition \ref{d.Nash}.

\smallskip

Having dealt with the case of extreme equilibrium, until the end of the proof we shall assume that \eqref{eq:beta_more_one_up_to_one} holds but \eqref{eq.condition_for_extreme} fails. Without loss of generality, let trader $0 \in I$ have the maximal pre-transaction beta: $\beta_i \leq \beta_0$ for all $i \in I \setminus \set{0}$. In view of Lemma \ref{lem:condition_for_extreme}, we then have that, necessarily,
\begin{equation} \label{eq:cond}
-1 < \beta_0 <  1 + \frac{1}{\delta_0 } \sum_{\iii \setminus \set{0}} \delta_i (1 + \beta_i)_+.
\end{equation}
Define the set
\[
J \dfn \set{\iii \setminus \set{0} \such -1 < \beta_i \leq 1}.
\]
The set $J_0 \dfn J \cup \set{0}$ contains all traders that will eventually submit demand functions with non-zero elasticity. Note that \eqref{eq:cond} implies that $J \neq \emptyset$; indeed, if $J = \emptyset$, then $\beta_0 = 1 - \sum_{i \in I \setminus \set{0}}\beta_i > 1$, and \eqref{eq:cond} would fail, since the quantity at the right-hand side would equal 1.

A Nash equilibrium exists if and only if $\rg_i = 0$ holds for all $i \in I \setminus J_0$, while 
\[
\pare{2 + \frac{\rg_I - \rgi}{\delta_i} } \frac{\rgi}{\rg_I} =  1 + \beta_i, \quad \forall \iij,
\]
following from \eqref{eq.Nash_couples}. Given $\rg_I > 0$, $\rg_i$ for $\iij$ satisfies the quadratic equation
\begin{equation}\label{eq.quadratic_eq_proof}
\frac{1}{2} (\rg_i)^2 - \pare{\delta_i + \rg_I / 2} \rgi + \delta_i ( 1 + \beta_i) \rg_I/2 = 0, \quad \forall \iij.
\end{equation}
The discriminant is equal to $\pare{\delta_i + \rg_I / 2}^2 - \delta_i ( 1 + \beta_i) \rg_I$, which, since $-1 < \beta_i \leq 1$, is always (regardless of the value of $\rg_I$) nonnegative. The two roots of equation \eqref{eq.quadratic_eq_proof} are $\delta_i + \rg_I / 2 \pm \sqrt{\pare{\delta_i + \rg_I / 2}^2 - \delta_i ( 1 + \beta_i) \rg_I}$. Note that since
\begin{align*}
\delta_i + \rg_I / 2 + \sqrt{\pare{\delta_i + \rg_I / 2}^2 - \delta_i ( 1 + \beta_i) \rg_I} &\geq \delta_i + \rg_I / 2 + \sqrt{\pare{\delta_i + \rg_I / 2}^2 - 2 \delta_i \rg_I}  \\
&= \delta_i + \rg_I / 2 + |\rg_I / 2 - \delta_i| \geq \rg_I,
\end{align*}
and $\rg_0$ has to be strictly positive, it holds that $\rg_i < \rg_I$ for each $i \in J$. Hence, the only root that is acceptable, i.e., the only nonnegative root is  
\[
\rg_i = \delta_i + \rg_I / 2 - \sqrt{\pare{\delta_i + \rg_I / 2}^2 - \delta_i ( 1 + \beta_i) \rg_I}, \quad \forall \iij.
\]
(Recall that our Definition of noncompetitive equilibrium considers linear demand functions with nonpositive slopes.) In other words, and upon defining the function $\phi_i : (0, \infty) \mapsto \Real$ via
\[
\phi_i(x) \dfn \delta_i + x/2 - \sqrt{\pare{\delta_i + x / 2}^2 - \delta_i ( 1 + \beta_i) x}, \quad x > 0,
\]
we should have $\rg_i = \phi_i(\rg_I)$ for all $\iij$. The next result gives some necessary properties on $\phi_i$ for $\iij$.

\begin{lem} \label{lem:phi_i}
	Let $\iij$. Then, $\phi_i(0+) = 0$, $\phi'_i(0+) = (1 + \beta_i)/2$. Furthermore, $\phi_i$ is concave, nondecreasing, and such that $\phi_i(\infty) = \delta_i (1 + \beta_i)$.
\end{lem}

\begin{proof}
	The fact that  $\phi_i(0+) = 0$ is immediate. In the special case $\beta_i = 1$, we have $\phi_i(x) = \delta_i + x/2 - \abs{x / 2 - \delta_i} = x \wedge (2 \delta_i)$ for $x > 0$, and the result is trivial. When $-1 < \beta_i < 1$, $\phi_i$ is twice continuously differentiable, and an easy calculation gives
	\[
	\phi'_i(x) = \frac{1}{2} - \frac{x / 2 - \delta_i \beta_i}{2 \sqrt{\pare{\delta_i + x / 2}^2 - \delta_i ( 1 + \beta_i) x}}, \quad x > 0,
	\]
	from which it immediately follows that $\phi'_i(0+) = (1 + \beta_i)/2$.
	Furthermore, another easy calculation gives
	\[
	\phi''_i(x) = \frac{- 1 + \pare{x / 2 - \delta_i \beta_i}^2 / \pare{\pare{\delta_i + x / 2}^2 - \delta_i ( 1 + \beta_i) x}}{2 \sqrt{\pare{\delta_i + x / 2}^2 - \delta_i ( 1 + \beta_i) x}}, \quad x > 0.
	\]
	Therefore, $\phi''_i(x) < 0$ for all $x > 0$ is equivalent to $\pare{x / 2 - \delta_i \beta_i}^2 < \pare{\delta_i + x / 2}^2 - \delta_i ( 1 + \beta_i) x$ for all $x > 0$. Calculating the squares and cancelling terms, we obtain $\delta_i^2 \beta_i^2 < \delta_i^2$, which is true since $-1 < \beta_i < 1$. Therefore, $\phi_i$ is concave. Continuing, a straightforward calculation gives
	\begin{align*}
	\frac{x}{2} - \sqrt{\pare{\delta_i + x / 2}^2 - \delta_i ( 1 + \beta_i) x} &= \frac{\pare{x /2}^2 - \pare{\pare{\delta_i + x / 2}^2 - \delta_i ( 1 + \beta_i) x}}{ x / 2 + \sqrt{\pare{\delta_i + x / 2}^2 - \delta_i ( 1 + \beta_i) x}} \\
	&= \frac{- \delta_i^2 + \delta_i \beta_i x}{ x / 2 + \sqrt{\pare{\delta_i + x / 2}^2 - \delta_i ( 1 + \beta_i) x}},
	\end{align*}
	which, as $x \to \infty$, has limit $\delta_i \beta_i$. Therefore, $\phi_i(\infty) = \delta_i (1 + \beta_i) > 0$. Since $\phi_i(0) = 0 < \delta_i (1 + \beta_i) = \phi_i(\infty)$ and $\phi_i$ is concave, we conclude that it is nondecrasing.
\end{proof}

Regarding trader $0 \in I$, since $\beta_0 > -1$, in equilibrium we should have
\[
\pare{2 + \frac{\rg_I - \rg_0}{\delta_0} } \frac{\rg_0}{\rg_I} =  1 + \beta_0.
\]
Note that $\rg_I - \rg_0 = \sum_{\iij} \rg_i = \sum_{\iij} \phi_i(\rg_I)$. Therefore, upon defining
\[
\sigma(x) \dfn \sum_{\iij} \phi_i(x), \quad x > 0,
\]
we should have
\[
\pare{2 + \frac{\sigma(\rg_I)}{\delta_0} } \frac{\rg_0}{\rg_I} =  1 + \beta_0,
\]
which immediately gives
\[
\rg_0 =  \frac{\pare{1 + \beta_0} \delta_0}{2 \delta_0 +\sigma(\rg_I) } \rg_I,
\]
Hence, in equilibrium, the following equation should hold for $\rg_I > 0$:
\[
\frac{\pare{1 + \beta_0} \delta_0}{2 \delta_0 +\sigma(\rg_I) } \rg_I + \sigma(\rg_I) = \rg_I.
\]
In other words, at equilibrium $\rg_I$ should solve the equation
\begin{equation} \label{eq:key}
\frac{\pare{1 + \beta_0} \delta_0}{2 \delta_0 +\sigma(x) } + \frac{\sigma(x)}{x}  = 1, \quad x > 0.
\end{equation}
By Lemma \ref{lem:phi_i}, it follows that the left-hand-side of equation \eqref{eq:key} is decreasing in $x > 0$. Its limit at $x = 0+$ is equal to
\[
\frac{1 + \beta_0}{2} + \sum_{\iij} \frac{1 + \beta_i}{2} = \frac{|J_0|}{2} + \frac{1}{2} \sum_{\iij_0} \beta_i.
\]
Since $|J_0| \geq 2$ (recall that $J \neq \emptyset$) and $\sum_{\iij_0} \beta_i \geq 1$ (by definition of $J$ and the fact that $\beta_I = 1$), the above limit is strictly greater than one. It follows that \eqref{eq:key} will have a (necessarily unique) solution if and only if the limit as $x \to \infty$ of the left-hand-side of \eqref{eq:key} is strictly less than one. In other words, and since $\sigma (\infty) = \sum_{\iij} \pare{1 + \beta_i} \delta_i$, it should hold that
\[
\pare{1 + \beta_0} \delta_0  < 2 \delta_0 +\sigma(\infty) = 2 \delta_0 + \sum_{\iij} \pare{1 + \beta_i} \delta_i ,
\]
which is exactly \eqref{eq:cond}.

The above discussion implies that a unique Nash equilibrium exists under the validity of \eqref{eq:beta_more_one_up_to_one} and failure of \eqref{eq.condition_for_extreme_bis}, completing the proof of Theorem \ref{thm:ex_and_un}.

\bibliographystyle{amsalpha}
\bibliography{references}
\end{document}